\documentclass[a4paper,onecolumn,accepted=2019-05-14]{quantumarticle}
\pdfoutput=1
\usepackage[utf8]{inputenc}
\usepackage[english]{babel}
\usepackage[T1]{fontenc}
\usepackage{hyperref}
\usepackage{amsmath}
\usepackage{amsmath}
\usepackage{amsfonts}
\usepackage{amssymb}
\usepackage{makeidx}
\usepackage{graphicx}
\usepackage{mathrsfs}
\usepackage{amsthm}
\usepackage{color}
\usepackage[numbers,sort&compress]{natbib}
\newtheorem{prop}{Proposition}
\newtheorem{lemma}{Lemma}
\newtheorem{coro}{Corollary}
\theoremstyle{remark}
\newtheorem{exm}{Example}

\theoremstyle{definition}
\newtheorem{definition}{Definition}
\def\states{\mathfrak{S}}
\def\dens{\mathfrak{D}}

\def\effects{\mathcal{E}}

\def\Ha{\mathcal{H}}

\def\conv{\text{conv}}

\def\Tr{\text{Tr}}
\usepackage{dsfont}
\def\I{\mathds{1}}

\def\<{\langle}
\def\>{\rangle}

\begin{document}
%
%
\title{On the properties of spectral effect algebras}

\author{Anna Jenčová}
\email{jenca@mat.savba.sk}
\affiliation{Mathematical Institute, Slovak Academy of Sciences, \v Stef\' anikova 49, Bratislava, Slovakia}
\orcid{0000-0002-4019-268X}

\author{Martin Plávala}
\email{martin.plavala@mat.savba.sk}
\affiliation{Mathematical Institute, Slovak Academy of Sciences, \v Stef\' anikova 49, Bratislava, Slovakia}
\orcid{0000-0002-3597-2702}

\begin{abstract}
The aim of this paper is to show that there can be either only one or uncountably many contexts in any spectral effect algebra, answering a question posed in [S. Gudder, Convex and Sequential Effect Algebras, (2018), arXiv:1802.01265]. We also provide some results on the structure of spectral effect algebras and their state spaces and investigate the direct products and direct convex sums of spectral effect algebras. In the case of  spectral effect algebras with sharply determining state space, stronger properties can be proved: the spectral decompositions are essentially unique, the algebra is sharply dominating and the set of its sharp elements is an orthomodular lattice. The article also contains a list of open questions  that might provide interesting future research directions.
\end{abstract}

\maketitle

\section{Introduction}
Spectrality plays an important role in quantum theory. On one hand spectral theorem for operators over Hilbert spaces is a well known result found in many textbooks, on the other hand various forms of spectrality are used in operational derivations of quantum theory. There is quite a number of works concernng this line of research (see e.g. \cite{Hardy-derivation, ChiribellaDArianoPerinotti-derivQT} or the recent paper \cite{vandeWetering-filters} and references therein), where quantum theory is derived  from a set of postulates,  and obtaining  the spectral decomposition of states and effects is a crucial step. Spectrality was also studied in the framework of generalized probability theory \cite{Gudder-spectralGPT, ChiribellaScandolo-diagonalization} and is also important for the definition of entropy in operational frameworks \cite{KimuraIshiguroFukui-entropies, KimuraNuidaImai-indecomposability, BarnumBarrettClarkLeiferSpekkensStepanikWilceWilke-entropy, BarnumBarrettKrummMuller-entropy}.

In the special case of quantum theory over a finite-dimensional complex Hilbert space $\Ha$,  self-adjoint operators such that $0 \leq A \leq \I$ are called effects, here $A \geq 0$ means that $A$ is positive semi-definite and $\I$ is the identity operator. It is well-known that the effects  are  used in description of quantum measurements \cite{HeinosaariZiman-MLQT}. In the abstract setting,  effect algebras were defined in \cite{FoulisBennett-effAlg} as a generalization of the set of Hilbert space  effects. Since then a variety of mathematical results has been developed in the field and it has attracted the interest of both mathematicians and physicist alike.

Since spectrality is one of the key features of quantum theory and effects are one of the building blocks of operational generalizations of quantum theory \cite{ChiribellaDArianoPerinotti-GPT}, it is clear that effect algebras possessing some form of spectral property (or simply spectral effect algebras) are important  from the viewpoint of foundations of quantum theory as we may hope that any other physical theory, such as quantum theory of gravity, might have some form of spectral properties.

Spectral effect algebras were recently introduced in \cite{Gudder-contexts} as convex effect algebras such that any element can be decomposed as a finite combination of indecomposable effects summing up to the identity. Such sets of effects are called contexts.
It was proved that a spectral effect algebra is classical if and only if it possesses a single context. Moreover,  necessary and sufficient conditions were obtained for the algebra to be  affinely isomorphic to the set of effects over a complex Hilbert space. 
In this case, the algebra must have  infinitely many contexts which corresponded to the rank-1 projective measurements.
 It was left as an open question in \cite{Gudder-contexts} whether the number of contexts in a spectral effect algebra can be finite and greater than one, with the author conjectured that this is not possible. The conjecture was essecntially  proved in \cite{vandeWetering-seqProdSpaces} for sequential effect algebras, in showing that these are affinely isomorphic to unit intervals in Euclidean Jordan algebras.

The aim of the present paper is to  present a proof of this conjecture for general spectral effect algebras. Along the way, we 
 take a closer look on some of the properties of  spectral effect algebras and formulate a new set of open questions, that might lead to interesting research directions.

The article is organized as follows: in Sec. \ref{sec:effAlg} we present the basic definitions and provide examples of convex effect algebras. In Sec. \ref{sec:properties} we introduce spectral effect algebras and prove some basic results on properties of spectral decompositions and  the structure of the state space. In Sec. \ref{sec:number} we prove the main result that a spectral effect algebra may contain only either one or uncountably many contexts. In Sec. \ref{sec:composition} we inspect two standard constructions with convex effect algebras: the direct product and the direct convex sum and we show that while the direct product of spectral effect algebras is again a spectral effect algebra, the direct convex sum of spectral effect algebras is not spectral. In Sec. \ref{sec:sharplyDetermining} we look at a special case of spectral effect algebras that have sharply determining state space and we show that in such setting we obtain a stronger spectral decomposition,  analogical to the result \cite[Proposition 18]{vandeWetering-filters}, moreover, the algebra is sharply dominating and the set of its sharp elements is an orthomodular lattice. The final Sec. \ref{sec:conclusion} contains the conclusions and open questions.

\section{Properties of convex effect algebras} \label{sec:effAlg}
In this section we are going to present definition and properties of convex effect algebras. We closely follow the definitions used in \cite{Gudder-contexts} with a slightly different notation,  more natural for linear effect algebras that are closely related to general probabilistic theories.
\begin{definition} \label{def:effAlg-def}
An effect algebra is a system $(E, 0, 1, +)$, where $E$ is a set containing at least one element, $0, 1 \in E$ and $+$ is a partial binary operation on $E$. Let $a, b \in E$, then we write $a+b \in E$ whenever $a+b$ is defined (and hence yields an element of $E$). Moreover we require that $(E, 0, 1, +)$ satisfies the following conditions:
\begin{itemize}
\item[(E1)] if $a+b \in E$ then $b+a \in E$ and $a+b = b+a$,
\item[(E2)] if $a+b \in E$ and $(a+b)+c \in E$, then $a+(b+c) \in E$, $a+(b+c) = (a+b)+c$,
\item[(E3)] for every $a \in E$ there is unique $a' \in E$ such that $a + a' = 1$, we usually denote $a' = 1-a$,
\item[(E4)] if $a + 1 \in E$, then $a = 0$.
\end{itemize}
\end{definition}

\begin{definition}
An effect algebra $E$ is convex if for every $a \in E$ and $\lambda \in [0, 1] \subset \mathbb{R}$ there is an element $\lambda a \in E$ such that for all $\lambda, \mu \in [0, 1]$ and $a, b \in E$ we have
\begin{itemize}
\item[(C1)] $\mu (\lambda a) = (\lambda \mu) a$,
\item[(C2)] if $\lambda + \mu \leq 1$, then $\lambda a  + \mu a \in E$ and $(\lambda + \mu) a  = \lambda a + \mu a$,
\item[(C3)] if $a + b \in E$, then $\lambda a + \lambda b \in E$ and $\lambda (a+b) = \lambda a + \lambda b$,
\item[(C4)] $1 a = a$.
\end{itemize}
\end{definition}
The above definition of convex effect algebra allows one to define convex combinations of the form $\lambda a + (1-\lambda) b$ for $\lambda \in  [0, 1]$.
\begin{definition}
Let $E$, $F$ be  effect algebras. A map $\phi: E \mapsto F$ is called additive if for $a, b \in E$, $a+b\in E$  we have:
  $\phi(a) + \phi(b) \in F$ and $\phi(a+b) = \phi(a) + \phi(b)$. An additive map such that $\Phi(1)=1$ is called a morphism.
A morphism is an isomorphism if it is surjective and for $a,b\in E$, $\phi(a) + \phi(b) \in F$ implies $a+b \in E$.
If $E$ and $F$ are convex effect algebras, then a morphism $\Phi$ is affine if
\begin{equation*}
\Phi(\lambda a)=\lambda\Phi(a)
\end{equation*}
for $a\in E$, $\lambda\in [0,1]$.
\end{definition}

\begin{definition}
A state on an effect algebra $E$ is a morphism $s: E \mapsto [0,1] \subset \mathbb{R}$. The set of states on an effect algebra will be denoted $\states(E)$.
\end{definition}
It was proved in \cite{Gudder-effectAlgebras} that any state on a convex effect algebra is affine.
The set of states is also referred to as state space of the effect algebra and it will play an important role in later constructions.

Let $V$ be a real vector space with a pointed convex cone $P$, that is $P \cap (-P) = \{ 0 \}$ where $0$ denotes the zero vector. For $v, w \in V$ we define $v \geq w$ if and only if $v - w \in P$. Then $\le$ is a partial order in $V$  and $(V,P)$ is an ordered vector space.
Let $u \in P$, then the set
\begin{equation*}
[0, u] = \{ v \in V: 0 \leq v \leq u \}
\end{equation*}
is an effect algebra with the operation $+$ defined as the sum of the vectors and for $a, b \in [0, u]$ we have $a+b \in [0, u]$ if and only if $a+b \leq u$ which is exactly why we have chosen such unusual notation in Def. \ref{def:effAlg-def}. Also note that in this case $u$ is the unit of the effect algebra $[0, u]$, i.e. we have $1 = u$.

\begin{definition}
A linear effect algebra is an effect algebra of the form $[0, u]$ for some ordered  vector space $(V,P)$ and  $u\in P$.
\end{definition}

The following is an important result.
\begin{prop}
Every convex effect algebra is affinely isomorphic to a linear effect algebra.
\end{prop}
\begin{proof}
See \cite{GudderPulmannova-representation}.
\end{proof}
From now on we are going to assume that all of the effect algebras we will work with are convex.  Below, we omit the isomorphism and  identify convex effect algebras with the linear effect algebras they are isomorphic to.
Moreover, we may and will assume that the  interval $[0,u]$ generates the ordered vector space $(V,P)$, so that $u$ is an order unit in $(V,P)$, \cite[Lemma 3.1]{GudderPulmannovaBugajskiBeltrametti-sharpEffects}.

\begin{definition}\label{def:sharpex}
Let $f \in E$. We say that
\begin{itemize}
\item  $f$ is one-dimensional if $f \neq 0$ and for $g \in E$ we have that $f \geq g$ implies $g = \lambda f$ for some $\lambda \in [0, 1]$;
\item $f$ is sharp if $f \geq g$ and $1-f \geq g$ implies $g = 0$;
\item $f$ is extremal if $f=\lambda g_1+(1-\lambda) g_2$ for some $\lambda\in (0,1)$ implies that $f=g_1=g_2$.
\end{itemize}
The set of sharp elements will be denoted by $S(E)$. The set of sharp one-dimensional elements  will be denoted by $S_1(E)$. 
\end{definition}

It was shown in \cite[Lemma 4.4]{GudderPulmannovaBugajskiBeltrametti-sharpEffects} that any extremal element is sharp, but the converse is not necessarily true. Moreover,  for any sharp effect $f$ there are states $s_0, s_1 \in \states(E)$ such that $s_0(f) = 0$, $s_1(f) = 1$. Note that in general the states $s_0$ and $s_1$ do not have to be unique. In general probabilistic theories \cite{KimuraNuidaImai-indecomposability} one-dimensional effects are called indecomposable and they play a key role in the results on simulation irreducible measurements \cite{FilippovHeinosaariLeppajarvi-GPTcompatibility}. Some authors also refer to sharp one-dimensional effects as atomic \cite{vandeWetering-seqProdSpaces}. In what follows we will be interested in the properties of sharp one-dimensional effects.
\begin{prop} \label{prop:effAlg-sharp1D}
Let $E$ be a convex effect algebra, then a one-dimensional effect is sharp if and only if it is extremal.
\end{prop}
\begin{proof}
As noted above, every extremal element is sharp. Conversely,
let $f\in S_1(E)$ and assume that we have $f = \lambda g_1 + (1-\lambda) g_2$ for some $g_1, g_2 \in E$ and $\lambda \in (0, 1)$. It follows that we have $f \geq \lambda g_1$ and $f \geq (1-\lambda) g_2$ and since $f$ is one-dimensional we must have $g_1 = \mu_1 f$ and $g_2 = \mu_2 f$, i.e. $f$ must be a convex combination of its multiples. Let now $\mu \ge  1$ be such that  $\mu f \in E$.
 Let $\eta = \min\{ \frac{1}{\mu}, 1-\frac{1}{\mu}\}$ then we have \begin{align*}
f &\geq \eta f \\
1-f = \dfrac{1}{\mu} (1 - \mu f) + (1-\dfrac{1}{\mu}) 1 \geq \eta 1 &\geq \eta f.
\end{align*}
Since $f$ is sharp, this implies that  $\eta=0$ and hence $\mu=1$. It follows that we must have $g_1=g_2=f$ and $f$ is extremal.
\end{proof}

\begin{exm}[Classical theory] \label{exm:effAlg-clasTheory}
Let $V_n$ be a $n$-dimensional real vector space and let $v \in V_n$ be given as
\begin{equation*}
v =
\begin{pmatrix}
v^1 \\
\vdots \\
v^n
\end{pmatrix}
\end{equation*}
where $v_i \in \mathbb{R}$ are the coordinates of the vector. The positive orthant is the cone
\begin{equation*}
P_n = \{ v \in V_n: v^i \geq 0, \forall i \in \{1, \ldots, n\} \},
\end{equation*}
one can show that $P_n$ is a pointed convex cone. Let
\begin{equation*}
u_n = \begin{pmatrix}
1 \\
\vdots \\
1
\end{pmatrix}
\end{equation*}
be the vector with all coordinates equal to $1$, then $[0, u_n] = E(S_n)$ is an effect algebra and we call it the classical effect algebra. Now we will find the state space of $E(S_n)$. Note that since the cone $P_n$ is also generating, i.e. $V_n = P_n - P_n$, every morphism $s: E(S_n) \to [0, 1]$ can be extended to a linear functional on $V_n$, i.e. we have $\states(E(S_n)) \subset V_n^*$ where $V_n^*$ is the dual vector space. Moreover note that $V_n^*$ is again $n$-dimensional vector space hence we can use the same coordinate representation for $w \in V_n^*$ as for $v \in V_n$ and we have $w(v) = \sum_{i=1}^n w^i v^i$.

Let $s_1, \ldots, s_n$ denote the standard basis of $V$, i.e. we have $s^i_j = \delta_{ij}$ where $\delta_{ij}$ is the Kronecker delta, then one can easily see that $\states(E(S_n)) = \conv(\{ s_1, \ldots, s_n\}) = S_n$ where $S_n$ denotes the standard simplex with $n$ vertexes.
\end{exm}

\begin{exm}[Quantum theory]
The most well-know example of an effect algebra comes from quantum theory. Let $\Ha$ be a finite-dimensional Hilbert space and let $B_h(\Ha)$ be the real vector space of self-adjoint operators on $\Ha$. There is a natural positive cone of positive semi-definite operators $B_h(\Ha)^+ = \{ A \in B_h(\Ha): A \geq 0 \}$ and a natural selection of unit effect to be the identity operator $\I$. But actually any positive invertible operator would generate some convex effect algebra in a similar fashion. The interval $\effects(\Ha) = [0, \I] = \{ E \in B_h(\Ha) : 0 \leq E \leq \I \}$ is a convex effect algebra and the corresponding operators are called effects. They are exactly the operators we use to form POVMs which represent measurements in quantum theory \cite{HeinosaariZiman-MLQT}. One can see that in $\effects(\Ha)$ the sharp elements are projections and sharp one-dimensional elements are projections on one dimensional subspaces of $\Ha$ (also called rank-1 projections). It is rather straightforward to find the state space of $\effects(\Ha)$ as we have $\states(\effects(\Ha)) = \dens_\Ha = \{ \rho \in B_h(\Ha): \rho \geq 0, \Tr(\rho) = 1 \}$, i.e. it is the set of density matrices.
\end{exm}

\begin{exm}[Square state space]
Let $V$ be a real finite dimensional vector space with $\dim(V) = 3$ and let us denote
\begin{align*}
&f = \begin{pmatrix}
1 \\
0 \\
0 \\
\end{pmatrix},
&g = \begin{pmatrix}
0 \\
1 \\
0 \\
\end{pmatrix},
&&u = \begin{pmatrix}
0 \\
0 \\
1 \\
\end{pmatrix}.
\end{align*}
Let $P$ be the convex cone generated by the vectors $f, g, u-f, u-g$, one can show that this cone is pointed. Let $[0, u] = E(S)$ be an effect algebra. The state space is given as
\begin{equation*}
\states(E(S)) = \{ v \in V: v^1, v^2 \in [0, 1], v^3 = 1 \}
\end{equation*}
and it is affinely isomorphic to a square $S = [0, 1] \times [0, 1]$.

Let $h = \frac{1}{2}( f + g)$, then clearly $h$ is not an extreme point of $E(S)$, we will show that it is sharp hence showing that in Prop. \ref{prop:effAlg-sharp1D} one has to require that the effect in question is one-dimensional. Let $s_{00}, s_{10}, s_{01}, s_{11} \in \states(E(S))$ be the extreme points given as
\begin{align*}
&s_{00} =
\begin{pmatrix}
0 \\
0 \\
1
\end{pmatrix}
&s_{10} =
\begin{pmatrix}
1 \\
0 \\
1
\end{pmatrix}
&&s_{01} =
\begin{pmatrix}
0 \\
1 \\
1
\end{pmatrix}
&&&s_{11} =
\begin{pmatrix}
1 \\
1 \\
1
\end{pmatrix}
\end{align*}
and let $h \geq p$ and $u-h \geq p$. Remember that the states are morphisms, i.e. for $s \in \states(E(S))$ we must have $s(h) \geq s(p)$ and $1-s(h) = s(u-h) \geq s(p)$. From $s_{00}(h) = 0$ we get $s_{00}(p) = 0$ and from $s_{11}(h) = 1$ we get $s_{11}(p) = 0$. Moreover note that
\begin{equation*}
s_{00} + s_{11} = s_{10} + s_{01}
\end{equation*}
so we must have
\begin{equation*}
s_{10} (p) = (s_{00} + s_{11})(p) - s_{01}(p) = - s_{01}(p).
\end{equation*}
From $s_{10} (p) \geq 0$ and $s_{01} (p) \geq 0$ it follows that $s_{10} (p) = s_{01} (p) = 0$ so $p$ is zero on all extreme points. In this case it follows that $p = 0$ as it is zero on all states.
\end{exm}

\begin{exm}[General probabilistic theories]
We have already seen that every effect algebra gives rise to a state space (even though in some cases it can be an empty set). One can reverse the logic and ask: can we use a state space to generate an effect algebra? The answer is yes and this approach is often termed general probabilistic theories.

Let $V$ be a finite-dimensional real vector space with the standard Euclidean metric and let $K \subset V$ be a compact convex set. Let $f: K \to \mathbb{R}$ be a function, then we say that $f$ is affine if for all $x,y \in K$ and $\lambda \in [0,1]$ we have
\begin{equation*}
f(\lambda x + (1-\lambda)y) = \lambda f(x) + (1-\lambda) f(y).
\end{equation*}
Let $E(K)$ be the set of affine functions $f: K \to \mathbb{R}$ such that for all $x \in K$ we have $f(x) \in [0,1]$, then one can see that $E(K)$ is an effect algebra with the partially defined operation of addition and the unit effect $1(x) = 1$. In this case one can show that $\states(E(K)) = K$ up to an affine isomorphism \cite[Theorem 4.3]{AsimowEllis} and this duality often plays an important role in many results. Also note that
classical theory is a special case of a general probabilistic theory with $K = S_n$, quantum theory is a special case of a general probabilistic theory with $K = \dens_\Ha$ and the square state space theory is a special case of a general probabilistic theory with $K = S$.
\end{exm}

\section{Contexts and spectral effect algebras} \label{sec:properties}
In this section we are going to introduce contexts and spectral effect algebras and provide some results on their structure.

\begin{definition}
A context in a convex effect algebra $E$ is a finite collection $A = \{a_1, \ldots, a_n\} \subset S_1(E)$ such that $\sum_{i=1}^n a_i = 1$.
\end{definition}

\begin{definition}
Let $A= \{ a_1, \ldots, a_n\}$ be a context in a convex effect algebra $E$ then we denote
\begin{equation*}
\hat{a}_i = \{ s_i \in \states(E): s_i (a_i) = 1 \}.
\end{equation*}
\end{definition}
By the remarks below Definition \ref{def:sharpex}, the sets $\hat{a}_i$ are nonempty and it is easy to see that $s_i (a_j) = \delta_{ij}$ for all $s_i\in \hat{a}_i$ and $i,j=1,\dots,n$, where $\delta_{ij}$ is the Kronecker delta.

Let us  denote by $\bar{A}$ the convex effect subalgebra generated by $A$, that is
\[
\bar{A}=\left\lbrace \sum_i\mu_ia_i,\ \mu_i\in [0,1] \right\rbrace .
\]
Then $\bar{A}$ is the interval $[0,u_A]$ in the ordered vector space $V(A):=\{\sum_it_ia_i,\ t_i\in \mathbb R\}$, with an obvious positive cone and order unit $u_A=\sum_i a_i$. 

\begin{definition}
We say that a convex effect algebra $E$ is spectral if for every $f \in E$ there is a context $A \subset E$ such that $f \in \bar{A}$. Specifically, any $f\in E$ has the form 
\begin{equation}\label{eq:spectral}
f=\sum_i \mu_i a_i
\end{equation}
for some $\mu_i\in [0,1]$ and some context $A=\{a_1,\dots,a_n\}$. Any expression of the form \eqref{eq:spectral} will be called a spectral decomposition of $f$.
\end{definition}

\begin{exm}[Classical theory]
In \cite[Theorem 3.1]{Gudder-contexts} it was shown that the classical effect algebras are exactly those effect algebras that contains only one context. Using the notation from Exm. \ref{exm:effAlg-clasTheory} let $a_1, \ldots, a_n \in E(S_n)$ be given as $a^i_j = \delta_{ij}$. It is easy to show that $a_1, \ldots, a_n$ are one-dimensional and extreme effect and we have $\sum_{i=1}^n e_i = u$, so $A = \{a_1, \ldots a_n \}$ is a context. Moreover note that $\bar{A} = E(S_n)$ so $E(S_n)$ is spectral.
\end{exm}

\begin{exm}[Circle state space]
Let $V = B_h(\mathbb{R}^2)$ be the vector space of $2 \times 2$ real symmetric matrices, let $P = B^+(\mathbb{R}^2)$ be the cone of symmetric positive semidefinite matrices and let $\I$ denote the identity matrix. We already know that $\effects(\mathbb{R}^2) = [0, \I]$ is an effect algebra. The state space of this effect algebra is affinely isomorphic to a circle.

It is easy to see that the sharp one-dimensional elements of $\effects(\mathbb{R}^2)$ are the rank-1 projection, i.e. the projections on the 1 dimensional subspaces of $\mathbb{R}^2$. It follows that every context is of the form $\{P, \I-P\}$; we are going to denote $\I - P = P^\perp$ to keep the equations short. Let $X \in \effects(\mathbb{R}^2)$ and let $P$ be a projection on one of the eigenvectors of $X$ then we have $X = \mu_1 P + \mu_2 (\I-P)$ where $\mu_i \in [0,1]$ for $i \in \{1, 2\}$ are the corresponding eigenvalues of $X$. This shows that $\effects(\mathbb{R}^2)$ is spectral.
\end{exm}

\begin{exm}[Quantum theory]
In quantum theory the contexts are the sets of the rank-1 projectors $P_i$, $i = 1, \ldots, n$ such that $\sum_{i=1}^n P_i = \I$. The effect algebra $\effects(\Ha)$ is spectral; this statement is just another formulation of the spectral theorem.
\end{exm}

Note that it is not assumed that the  number of elements is the same in each context or that the space $V$ generated by $E$ is finite dimensional, although it is the case in the prototypic examples.  These are the algebras of finite dimensional classical and quantum effects, that were characterized in \cite{Gudder-contexts}. It is also easy to see that the algebra of effects over a finite dimensional real Hilbert space, or more generally of effects in an Euclidean Jordan algebra, is spectral and has the above properties. 

Other examples are obtained from finite dimensional spectral convex sets introduced in  \cite{AlfsenSchultz-stateSpaces}. It can be seen from  \cite[Theorem 8.64]{AlfsenSchultz-stateSpaces} that the algebra of effects (that is, affine functions into $[0,1]$) over a spectral convex set is a spectral effect algebra. In particular, \cite[Theorem 8.87 and Fig. 8.1]{AlfsenSchultz-stateSpaces} provides a set of more exotic examples that are not of the obvious kind listed in the previous paragraph.

We will discuss some properties of spectral effect algebras and their state spaces. The following will be a useful tool.

\begin{lemma} \label{lemma:properties-minMax}
Let $f = \sum_{i=1}^n \mu_i a_i$ be a spectral decomposition of $f$. Then
\begin{align*}
\max_{s \in \states(E)} s(f) &= \max\{ \mu_1, \ldots, \mu_n \}, \\
\min_{s \in \states(E)} s(f) &= \min\{ \mu_1, \ldots, \mu_n \}.
\end{align*}
\end{lemma}
\begin{proof}
Let $s \in \states(E)$, then we have
\begin{equation*}
s (f) = \sum_{i=1}^n \mu_i s(a_i)
\leq \max\{ \mu_1, \ldots, \mu_n \} \sum_{i=1}^n s(a_i)
= \max\{ \mu_1, \ldots, \mu_n \}.
\end{equation*}
To show that the bound is tight let $\max\{ \mu_1, \ldots, \mu_n \} = \mu_{i'}$ for some $i' \in \{ 1, \ldots, n \}$ and let $s\in \hat{a}_{i'}$.  Then we have
\begin{equation*}
s(f) = \mu_{i'}=\max\{ \mu_1, \ldots, \mu_n \}.
\end{equation*}
The proof for the minimum is analogical.
\end{proof}

We next show that  the elements of the embedding ordered vector space $V$ have spectral decompositions as well.

\begin{lemma}\label{lemma:decV}
Let $E$ be spectral and let $(V,P)$ be the ordered vector space with an order unit $u\in P$ such that $E=[0,u]$. Then for any $v\in V$ there is a context $A$  such  that $v\in V(A)$.
\end{lemma}

\begin{proof}
Since $u$ is a order unit, we have $-\lambda u\le v\le \lambda u$ for some $\lambda> 0$. Then $0\le v+\lambda u\le 2\lambda u$, so that 
$a:=\frac1{2\lambda}v +\frac12 u\in E$.  Let $a=\sum_i \nu_i a_i$ for some context $A$, then $v=\lambda(2a-u)=\sum_i\lambda(2\nu_i-1)a_i$.
\end{proof}

\begin{prop}\label{prop:archimedean}
Any spectral effect algebra $E$ has an order determining set of states, that is,
$s(a)\le s(b)$ for all $s\in \states(E)$ implies that $a\le b$.
\end{prop}

\begin{proof}
Let $a,b\in E$ be such that $s(a)\le s(b)$ for all states $s$. By Lemma \ref{lemma:decV}, there is some context $A$ and real numbers $\mu_i$ such that 
 $v=b-a=\sum_i\mu_ia_i$. By the assumption, $\mu_i =s_i(v)\ge 0$ for $s_i\in \hat{a}_i$, so that $b-a\ge 0$ and $a\le b$. 
\end{proof}

Let us remark that it was proved in  \cite[Theorem 3.6]{GudderPulmannovaBugajskiBeltrametti-sharpEffects} that a convex effect algebra 
 $E$ has an order determining set of states if and only if it is Archimedean. In this case, the order unit seminorm 
 \[
\|v\|:=\inf\{\lambda>0,\ -\lambda u\le v\le \lambda u\}=\sup_{s\in \states(E)} |s(v)|
 \]
is a norm. Using Lemma \ref{lemma:properties-minMax}, we obtain for $f\in E$ with spectral decomposition  $f=\sum_i \mu_i a_i$ that
\begin{align}
&\Vert f \Vert=\max\{\mu_1,\ldots,\mu_n\},
&\Vert 1-f \Vert=1-\min\{\mu_1,\ldots,\mu_n\}. \label{eq:norm}
\end{align}
More generally, if $v=\sum_i \alpha_i a_i$ is a spectral decomposition, then we have for $s_i\in \hat{a}_i$
\[
\|v\|=\sup_{s\in \states(E)} |s(v)|=\max_i|\alpha_i|=:|\alpha_{i_{max}}|=|s_{i_{max}}(v)|.
\]

An element $f\in E$   does not have to have a  unique spectral decomposition. 
The following result is inspired by \cite{vandeWetering-filters} and is immediate from \eqref{eq:norm}.
\begin{prop} \label{prop:properties-decompositionMinMax}
Let $f \in E$ have two spectral decompositions 
\begin{equation*}
\sum_{i=1}^n \mu_i a_i = f = \sum_{j=1}^m \nu_j b_j
\end{equation*}
Then  $\max \{ \mu_1, \ldots, \mu_n \} = \max \{ \nu_1, \ldots, \nu_m \}$ and $\min \{ \mu_1, \ldots, \mu_n \} = \min \{ \nu_1, \ldots, \nu_m \}$.
\end{prop}

Up to now it is not clear whether the sets $\hat{a}_i$ consist of a single state. Clearly, these sets are faces of $\states(E)$. More generally, for any element  $f\in E$, the set $\hat{f}:=\{s\in \states(E), s(f)=1\}$ is a face of $\states(E)$ (note that this also may be empty).
\begin{definition}
Let $\mathcal{F} \subset \states(E)$ be a face, then we say that $\mathcal{F}$ is an $E$-exposed face if there is $f \in E$ such that $\mathcal{F} = \hat{f} = \{s\in \states(E), s(f)=1\}$. If $\{s\} = \mathcal{F} =\hat{f}$, we say that $s$ is and $E$-exposed point of $\states(E)$, in this case, we write $\hat{f}=s$.
\end{definition}
If $a\in S_1(E)$, $\hat{a}$ is a  nonempty $E$-exposed proper face of $\states(E)$. 
Note that if   $A = \{ a_1, \ldots, a_n\}$ is a context,  the faces $\hat{a}_i$ are affinely independent. 
We next show a property of the $E$-exposed points of $\states(E)$.

\begin{prop} \label{prop:properties-exposedStates}
Every $E$-exposed point of $\states(E)$ has the form $\hat{a}$ for some $a\in S_1(E)$.
\end{prop}
\begin{proof}
Let $s \in \states(E)$ be an $E$-exposed point, so that $s=\hat{f}$ for some  $f \in E$. Then $1=s(f)= \max_{s' \in \states(E)} s' (f)$ and $s$ is the unique point 
 where this maximum is attained.  We have already showed in the proof of Lemma \ref{lemma:properties-minMax} that every effect in $E$ attains its maximum at some 
$s\in \hat{a}$ with $a\in S_1(E)$. The result follows.
\end{proof}

\section{Number of contexts in a spectral effect algebra} \label{sec:number}
In this section we are going to prove the main result that answers the open question from \cite{Gudder-contexts} about the possible number of contexts in a spectral effect algebra.
\begin{prop} \label{prop:number-disjoint}
Assume that every pair   $a, a' \in S_1(E)$ is summable, i.e. $a + a' \in E$.  Then there is only one context in $E$.
\end{prop}
\begin{proof}
Assume that there are two contexts $A\ne B$  in $E$. Then there is at least one effect $a$ such that $a \in A$ but $a \notin B$. Let $B = \{b_1, \ldots b_n\}$ and let $s\in \hat{a}$. For every $b_i$, $i \in \{1, \ldots, n\}$, we must have $a + b_i \in E$. 
This implies
\begin{equation*}
s(a + b_i) \leq 1
\end{equation*}
which gives $s(b_i) = 0$ for all $i \in \{1, \ldots, n\}$ which is a contradiction with $\sum_{i=1}^n b_i = 1$.
\end{proof}

\begin{prop}
Every spectral effect algebra $E$ contains either one context or uncountably many contexts.
\end{prop}
\begin{proof}
Assume that $E$ contains at least two contexts. As a result of Prop. \ref{prop:number-disjoint},   there are  
$a, b \in S_1(E)$, $a\ne b$, that are not summable. Let $\lambda \in [0, 1]$ and denote
\begin{equation*}
c_\lambda = \lambda a + (1-\lambda) b.
\end{equation*}
We will show that every $c_\lambda$ must belong to a different context, hence $E$ we must contain uncountably infinite number of contexts.

Since $E$ is spectral we have that for every $\lambda \in [0, 1]$ there is a context $C_\lambda$ such that $c_\lambda \in \bar{C}_\lambda$. Assume that for some $\lambda \neq \mu$, $\mu \in [0, 1]$ we have $C_\lambda = C_\mu$. Then we have $a, b \in V(C_\lambda)$ as $c_\lambda, c_\mu \in V(C_\lambda)$ and we can express $a$ and $b$ as linear combinations of $c_\lambda$ and $c_\mu$.

Let $C_\lambda = \{ c_1, \ldots, c_n \}$, then we must have
\begin{equation*}
a = \sum_{i=1}^n \alpha_i c_i
\end{equation*}
for some $\alpha_i \in \mathbb{R}$. Moreover, let $s_i\in \hat{c}_i$, then $\alpha_i =s_i (a)  \in [0,1]$ for all $i$  implies that $a \in \bar{C}_\lambda$, which yields $a \in C_\lambda$ as $a$ is sharp and one-dimensional. In a similar fashion we get $b \in C_\lambda$, which is a contradiction with the assumption that $a$ and $b$ are not summable.
\end{proof}

\section{Compositions of spectral effect algebras} \label{sec:composition}
In the study of spectral effect algebras, it is a natural question  whether spectrality is preserved by some constructions over convex effect algebras. In this section, we study the  direct products and direct convex sums; note that these are the product and coproduct in the category of convex effect algebras.

Let $E_1$, $E_2$ be effect algebras and let $E_1 \times E_2 = \{ (f_1, f_2) : f_1 \in E_1, f_2 \in E_2 \}$. We are going to show that $E_1 \times E_2$ is a convex effect algebra. The partial binary operation $+$ is given for $f_i, f'_i \in E_i$, such that $f_i + f'_i \in E_i$, $i \in \{1, 2\}$ as $(f_1, f_2) + (f'_1, f'_2) = (f_1 + f'_1, f_2 + f'_2)$. The unit of $E_1 \times E_2$ is $(1, 1)$. If $E_1$, $E_2$ are convex effect algebras, then we can define a convex structure on $E_1 \times E_2$ by $\lambda (f_1, f_2) = (\lambda f_1, \lambda f_2)$, for $\lambda \in [0, 1]$. It follows that $E_1 \times E_2$ is a convex effect algebra.
\begin{definition}
Let $E_1$, $E_2$ be effect algebras, then their direct product is the effect algebra $E_1 \times E_2$.
\end{definition}
In the less general circumstances of the general probabilistic theories the direct product of the effect algebras can be introduced as the effect algebra corresponding to the direct convex sum of state spaces \cite[Proposition 1]{HeinosaariLeppajarviPlavala-noFreeInformation}.

\begin{exm}[Direct product of classical effect algebras]
We are going to construct the direct product $E(S_n) \times E(S_m)$. The vector space $V_{n+m}$ containing $E(S_n) \times E(S_m)$ is $(n+m)$-dimensional and the positive cone $P_{n+m}$ is again the positive orhant in $V_{n+m}$.The unit effect is the vector $u_{n+m} = (u_n, u_m)$, i.e. the vector with all coordinates equal to $1$. It follows that we have $E(S_n) \times E(S_m) = E(S_{n+m})$.
\end{exm}

\begin{prop} \label{prop:composition-directProduct}
The direct product of spectral effect algebras is a spectral effect algebra.
\end{prop}
\begin{proof}
It is easy to see that an element  $(f_1, f_2) \in E_1 \times E_2$ is sharp if and only if both $f_1$ and $f_2$ are sharp. Moreover, since  $(f_1, f_2) = (f_1, 0) + (0, f_2)$, such an element is one-dimensional if and only if one of the elements is one-dimensional and the other is 0.  
Let now $A = \{ a_1, \ldots, a_n \} \subset E_1$ and $B = \{ b_1, \ldots, b_m \} \subset E_2$ be contexts.  It is straightforward to see that $(a_1, 0) + \ldots + (a_n, 0) + (0, b_1) + \ldots + (0, b_m) = (1, 1)$, hence
\begin{equation*}
\{ (a_1, 0), \ldots, (a_n, 0), (0, b_1), \ldots, (0, b_m) \} \subset E_1 \times E_2
\end{equation*}
is a context. Moreover, any context in $E_1\times E_2$ is of this form.

Finally let $(f_1, f_2) \in E_1 \times E_2$ be any element. Since $E_1$ and $E_2$ are spectral there are contexts $A \subset E_1$ and $B \subset E_2$ such that $f_1 \in \bar{A}$ and $f_2 \in \bar{B}$. It follows that we have
\begin{equation*}
(f_1, f_2) = \sum_{i=1}^n \mu_i (a_i, 0) + \sum_{j=1}^m \nu_j (0, b_j)
\end{equation*}
which shows that $E_1 \times E_2$ is spectral effect algebra.
\end{proof}

The definition of the direct convex sum is a bit more involved. Let $E_1$, $E_2$ be convex effect algebras, then they are affinely isomorphic to the intervals $[0, u_1] \subset C_1 \subset V_1$ and $[0, u_2] \subset C_2 \subset V_2$ respectively, where for $i \in \{1, 2\}$ we have that $V_i$ are vectors spaces, $C_i$ are pointed cones and $u_i \in C_i$. Take the vector space $V_1 \times V_2$ that corresponds to the coproduct of the respective vector spaces and define a relation of equivalence $\approx$ by $(u_1, 0) = (0, u_2)$, i.e. $(x_1, y_2) \approx (x_2, y_2)$ if and only if we have
\begin{equation*}
(x_1, y_1) = (x_2, y_2) + \alpha (u_1, -u_2)
\end{equation*}
for some $\alpha \in \mathbb{R}$. It is clear that $\approx$ is reflexive, symmetric and transitive. Let $V$ be the quotient space 
$V= (V_1 \times V_2)/_\approx$ and define the cone  $C = (C_1 \times C_2)/_\approx \subset V$ as the set of equivalence classes containing an element of $C_1\times C_2$. Clearly, $C$ is a convex cone containing $u:=[(u_1,0)]_\approx= [(0,u_2)]_\approx$. To see that $C$ is pointed, simply note that  $[(x, y)]_\approx \in C \cap (-C)$ if and only if we can choose $x \in C_1$, $y \in C_2$ and for some $\alpha \in \mathbb{R}$ we have $x + \alpha u_1 \in (- C_1)$ and $y - \alpha u_2 \in (-C_2)$. 
This implies $0\le x\le -\alpha u_1$ in $(V_1,C_1)$ so that $\alpha\le 0$. At the same time, $0\le y\le \alpha u_2$ in $(V_2,C_2)$, so that $\alpha\ge 0$. This implies $\alpha=0$ and consequently $(x,y)=(0,0)$. 
\begin{definition}
The direct convex sum of effect algebras $E_1$ and $E_2$ is
\begin{equation*}
E_1 \oplus E_2:= [0, u] \subset C.
\end{equation*}
\end{definition}
We can see that $E_1 \oplus E_2$ is a set of equivalence classes of the form $[(\lambda f_1,(1-\lambda)f_2)]_\approx$, where 
$f_1\in E_1$, $f_2\in E_2$ and $\lambda\in [0,1]$. 
Note  that $E_1 \oplus E_2$ is a convex effect algebra by construction as it is an interval in an ordered vector space.

\begin{exm}[Direct convex sum of classical effect algebras]
We are going to construct the direct convex sum $E(S_2) \oplus E(S_2)$. Let $V_2$ be the vector space such that $E(S_2) \subset V_2$ and $\dim(V_2) = 2$ and let $P_2$ be the cone generated by $E(S_2)$. Let $e_1, e_2 \in E(S_2)$ be given as
\begin{align*}
&e_1 =
\begin{pmatrix}
1 \\
0
\end{pmatrix}
&&e_2 =
\begin{pmatrix}
0 \\
1
\end{pmatrix}
\end{align*}
then the cone $P_2$ is generated by $e_1$ and $e_2$, i.e. $P_2$ is the set of all positive multiples of convex combinations of $e_1$ and $e_2$. The cone $P_2 \times P_2$ is clearly generated by the vectors $(e_1, 0)$, $(e_2, 0)$, $(0, e_1)$ and $(0, e_2)$  so the cone $P = (P_2 \times P_2) \setminus_\approx$ must be generated by the equivalence classes of said vectors. Let us denote $u = [(u_2, 0)]_\approx = [(0, u_2)]_\approx$, $f = [(e_1, 0)]_\approx$ and $g = [(0, e_1)]_\approx$. It follows that $u-f = [(e_2, 0)]_\approx$ and $u-g = [(0, e_2)]_\approx$. Note that it follows that $f, g, u-f, u-g$ are linearly dependent while $(e_1, 0), (e_2, 0), (0, e_1), (0, e_2)$ are linearly independent. Since the cone $P$ is generated by the vectors $f, g, u-f, u-g$ we see that it is affinely isomorphic to the cone generated by the effect algebra $E(S)$ corresponding to the square state space. Moreover one can see that $E(S_2) \oplus E(S_2)$ is affinely isomorphic to $E(S)$.
\end{exm}

In a similar fashion to the aforementioned \cite[Proposition 1]{HeinosaariLeppajarviPlavala-noFreeInformation} one can again see that in the framework of general probabilistic theories the direct convex sum of effect algebras can be constructed as the effect algebra corresponding to the direct product of state spaces.

\begin{prop} \label{prop:composition-directSum}
The direct convex sum of spectral effect algebras is not a spectral effect algebra.
\end{prop}
\begin{proof}
In a sense we are going to mimic the proof of Prop. \ref{prop:composition-directProduct} with the only difference that now we will show that there are two types of contexts on $E_1 \oplus E_2$: they are either of the form $\{ [(a_1, 0)]_\approx, \ldots, [(a_n, 0)]_\approx \}$ where $A = \{ a_1, \ldots a_n \} \subset E_1$ is a context, or  $\{ [(0, b_1)]_\approx, \ldots, [(0, b_m)]_\approx \}$ where $B = \{ b_1, \ldots b_m \} \subset E_2$ is a context.

Let $f=[(\lambda f_1, (1-\lambda) f_2)]_\approx \in S_1(E_1 \oplus E_2)$, then from
\begin{equation*}
[(\lambda f_1, (1-\lambda) f_2)]_\approx = \lambda [(f_1, 0)]_\approx + (1-\lambda) [(0, f_2)]_\approx
\end{equation*}
we see that both $\lambda [(f_1, 0)]_\approx$ and $(1-\lambda) [(0, f_2)]_\approx$ must be multiples of some $f$. This implies that there are some $t,\alpha\in [0,1]$ such that
\begin{align*}
&t\lambda f_1=\alpha u_1, &(1-t)(1-\lambda) f_2=\alpha u_2.
\end{align*}
Assuming that both $f_1$ and $f_2$ are nonzero, this implies that either both are multiples of identity or $\lambda\in \{0,1\}$. 
In both cases, $f$ is of the form  $f=[g_1,0]_\approx$ for some $g_1\in E_1$ or $f=[0,g_2]_\approx$ for some $g_2\in E_2$.
It is clear that $[(g_1,0)]\in S_1(E_1\oplus E_2)$  if and only if $g_1 \in S_1(E_1)$ and similarly for elements of the form $[(0,g_2)]_\approx$. 
From the definition, we can see that in this  case, $[(g_1,0)]_\approx$ and $[(0,g_2)]_\approx$ are not summable in $E_1 \oplus E_2$, hence these cannot belong to the same context.

It follows that  all contexts on $E_1 \oplus E_2$ are of the above two types. It is straightforward to see that if neither of 
 $f_1$ or $f_2$ is a multiple of identity and $\lambda \in (0, 1)$, the element $[(\lambda f_1, (1-\lambda) f_2)]_\approx $ cannot be given as  $\sum_{i=1}^n \mu_i [(a_i, 0)]_\approx$ or $\sum_{j=1}^m \nu_j [(0, b_j)]_\approx$.  Hence $E_1 \oplus E_2$ is not spectral.
\end{proof}

\section{Sharply determining state spaces}\label{sec:sharplyDetermining}
Here we consider a special case of spectral effect algebras, for which stronger properties can be proved.

\begin{definition}\label{def:sharply_determining}
We say that the state space $\states(E)$ is sharply determining if for any sharp $f \in E$ and any $g \in E$ such that $g \ngeq f$ there is a state $s \in \states(E)$ such that $s(f) = 1 > s(g)$.
\end{definition}

We first obtain  a way stronger version of Prop. \ref{prop:properties-decompositionMinMax} that is similar to \cite[Proposition 18]{vandeWetering-filters}.

\begin{prop} \label{prop:sharplyDetermining-decomposition}
Let $E$ be a spectral effect algebra such that $\states(E)$ is sharply determining. Let $A, B \subset E$ be contexts, $A = \{a_1, \ldots, a_n\}$ and $B = \{b_1, \ldots b_m\}$ such that for some $f \in E$ we have $f \in \bar{A} \cap \bar{B}$, specifically
\begin{equation*}
\sum_{i=1}^n \mu_i \left( \sum_{j_i = 1}^{n_i} a_{i_j} \right) = f = \sum_{k=1}^m \nu_k \left( \sum_{l_k = 1}^{m_k} b_{l_k} \right)
\end{equation*}
where $\mu_1 > \ldots > \mu_n > 0$ and $\nu_1 > \ldots > \nu_m > 0$. Then we have $n = m$, $\mu_i = \nu_i$ and $\sum_{j_i = 1}^{n_i} a_{i_j} = \sum_{l_i = 1}^{m_i} b_{l_i}$.
\end{prop}
\begin{proof}
Denote $a'_i = \sum_{j_i = 1}^{n_i} a_{i_j}$ and $b'_j = \sum_{l_k = 1}^{m_k} b_{l_k}$. As a result of Prop. \ref{prop:properties-decompositionMinMax} we already know that $\mu_1 = \nu_1$. We will only show that $a'_1 = b'_1$, the result will follow by repeating the same procedure for $f - \mu_1 a'_1 = f - \nu_1 b'_1$.

As first note that $a'_i$ and $b'_i$ are sharp \cite[Theorem 4.8]{GudderPulmannovaBugajskiBeltrametti-sharpEffects}. Assume that $a'_1 \ngeq b'_1$ then there is $s \in \states(E)$ such that $s(b'_1) = 1 > s(a'_1)$. We have $s(f) = \sum_{k=1}^m \nu_k s(b'_k) = \nu_1 = \mu_1$ as well as $s(f) = \sum_{i=1}^n \mu_i s(a'_i) < \mu_1$ which is a contradiction. Hence we must have $a'_1 \geq b'_1$ and by the same logic we must have also $b'_1 \geq a'_1$ which together yields $a'_1 = b'_1$.
\end{proof}

It was proved in \cite[Theorem 4.8]{GudderPulmannovaBugajskiBeltrametti-sharpEffects} that if $\states(E)$ is sharply determining, then all sharp elements are extremal, moreover, the sum of sharp elements, if it exists, is sharp. If  $E$ is also  spectral, then it is clear that the sharp elements are precisely the finite sums of one-dimensional sharp elements. The next results show that in this case $E$ is sharply dominating, see \cite{Gudder-sharplyDeterminingEffectAlgebras} for a definition, and the set $S(E)$ of all sharp elements with ordering induced from $E$ is an orthomodular lattice. We will need the following reformulation of the condition in Definition \ref{def:sharply_determining}. For the proof, it is enough to realize that $f$ is sharp if and only if $1-f$ is sharp. 

\begin{lemma}\label{lemma:sharply_determining} Let $E$ be a convex effect algebra. Then $\states(E)$ is sharply determining if and only if for any sharp $f \in E$ and any $g \in E$ such that $g \nleq f$ there is a state $s \in \states(E)$ such that $s(f) = 0 < s(g)$.
\end{lemma}

\begin{prop} Let $E$ be a spectral effect algebra with a sharply determining state space and let $a=\sum_i \mu_ia_i$ be a spectral decomposition.  Then $a^0:=\sum_{i,\mu_i>0} a_i$ is the smallest sharp element larger than $a$.
\end{prop}

\begin{proof} It is clear that $a^0$ is a sharp element and $a\le a^0$. Assume next that $b$ is a sharp element such that $a\le b$. Then for any $s\in \states(E)$ such that $s(b)=0$ we have $s(a)=\sum_i\mu_is(a_i)=0$, so that $s(a_i)=0$ whenever $\mu_i>0$. It follows that $s(a^0)=0$. By Lemma \ref{lemma:sharply_determining}, $a^0\le b$.
\end{proof}

\begin{prop}\label{prop:lattice} Let $E$  be a spectral effect algebra with a sharply determining state space. Let $f,g\in E$ be sharp. Then $(\lambda f+(1-\lambda)g)^0$ does not depend on $\lambda\in (1,0)$ and we have $(\lambda f+(1-\lambda)g)^0=f\vee g$ in $S(E)$. Moreover, we have $f\wedge g=1-((1-f)\vee (1-g))$.
\end{prop}

\begin{proof} Let us first observe that for any $f\in E$ and $s\in \states(E)$,  $s(f)=0$ iff $s(f^0)=0$, this is easy to see from the definition of $f^0$. Let $f,g$ be sharp and for $\lambda\in (0,1)$ let $p_\lambda:=(\lambda f+(1-\lambda)g)^0$. Then for $s\in \states(E)$, $s(p_{1/2})=0$ iff $s(\frac12(f+g))=0$ iff $s(f)=s(g)=0$ iff $s(\lambda f+(1-\lambda)g)=0$ iff $s(p_\lambda)=0$. By Lemma \ref{lemma:sharply_determining} this implies that $p_\lambda=p_{1/2}$ for all $\lambda\in (0,1)$. If $h\in E$ is any sharp element such that $f,g\le h$, then $s(h)=0$ implies $s(f)=s(g)=0$ so that $s(p_{1/2})=0$, hence $p_{1/2}\le h$. By \cite[Corollary 4.10]{GudderPulmannovaBugajskiBeltrametti-sharpEffects}, $h-p_{1/2}\in  S(E)$,  this shows that $p_\lambda=p_{1/2}=f\vee g$ in $ S(E)$. The last assertion follows by de Morgan laws.
\end{proof}

\begin{coro} Let $E$  be a spectral effect algebra with a sharply determining state space. Then $ S(E)$ is an orthomodular lattice.
\end{coro}

\begin{proof} By \cite[Corollary 4.9]{GudderPulmannovaBugajskiBeltrametti-sharpEffects}, $ S(E)$  is a sub-effect algebra in $E$ that is an orthoalgebra and by 
Proposition \ref{prop:lattice}, $ S(E)$ is a lattice. By \cite[Prop. 1.5.8]{DvurecenskijPulmannova-quantumStructures}, any lattice ordered orthoalgebra is an orthomodular lattice.
\end{proof}

\section{Conclusions and open questions} \label{sec:conclusion}
In this article we have proved that there can be either only one or uncountably many contexts contained in a spectral effect algebra as well as few other results concerning spectral effect algebras. We believe that spectral effect algebras provide  a suitable framework for the study of operational theories possessing some king of spectral properties and therefore are an interesting research direction for quantum foundations. In any case,  there are still quite a few open questions that will be left for future work.

\begin{enumerate}
\item What are the extreme points of a state space of a spectral effect algebra? In Prop. \ref{prop:properties-exposedStates} we have shown that all exposed points of the state space of a spectral effect algebra have the form $s=\hat{a}$ for $a\in S_1(E)$. It is an open question whether one can show a similar result for the extreme points of the state space.

\item Let $A=\{a_1,\dots,a_n\}$ be a context and let $\hat {A}=\conv(\cup_i\hat{a}_i)$. Since $\hat{a}_i$ are affinely independent, $\hat{A}=\oplus_c\hat{a}_i$ is their convex direct sum. Let $I$ be a set that indexes all of the contexts $A_\alpha$ of an effect algebra $E$; does it hold that $\states(E) = \cup_{\alpha \in I} \hat{A}_\alpha$? If  $\hat{a}$ are $E$-exposed points for all $a\in S_1(E)$, this would be even a stronger result than the one proposed in the question above.
 
\item When is the cardinality of contexts always the same? It is straightforward to see that all contexts in the effect algebras used in quantum theory have the same number of elements. It would be very interesting to know when this is true in general.

\item Is it possible to extend the results of Prop. \ref{prop:properties-decompositionMinMax} to a result similar to Prop. \ref{prop:sharplyDetermining-decomposition}? The property in question is to show that 
$\sum_{i=1}^n \mu_i a_i = f = \sum_{j=1}^m \nu_j b_j$ where $A = \{ a_1, \ldots, a_n \}$ and $B = \{ b_1, \ldots, b_m\}$ are contexts would imply that for every $\mu_i$ there is $\nu_j$ such that $\mu_i = \nu_j$. A stronger version of said result would be to show that Prop. \ref{prop:sharplyDetermining-decomposition} holds for all spectral effect algebras; this can be either done by showing that the state space of every spectral algebra is sharply determining or by other means.

\item For $a\in \states_1(E)$, does  $\hat{a}$ contain only one state? It is rather easy to find convex effect algebras where the sets $\hat{a}$ are not singletons, one can for example consider the effect algebra corresponding to the square state space (but note that this effect algebra is not spectral).

\item Does spectrality of the effect algebra implies any kind of weak duality between the effect algebra and the state space? If all of the sets $\hat{a}_i$ would contain only single point then it would be tempting to define a map $T: a_i \mapsto \hat{a}_i$ but it is of question whether the map would be well defined and affine. This would be a very strong result to prove or to at least find some conditions for when it holds.
\end{enumerate}

\begin{acknowledgments}
The authors are thankful to Leevi Lepp\"{a}j\"{a}rvi for helpful discussions in the early stages of the project, to John van de Wetering for helpful comments on the manuscript and for pointing out the results in \cite{AlfsenSchultz-stateSpaces} and to the editors and reviewers at Quantum for remarks on the readability of the manuscript. This research was supported by grant VEGA 2/0069/16 and by the grant of the Slovak Research and Development Agency under contract APVV-16-0073. MP acknowledges that this research was done during a PhD study at Faculty of Mathematics, Physics and Informatics of the Comenius University in Bratislava.
\end{acknowledgments}

\bibliographystyle{unsrtnat}
\bibliography{citations}

\end{document}